\documentclass{article}
\usepackage[utf8]{inputenc}

\usepackage{amsmath,amssymb,amsthm}
\usepackage{bm}
\usepackage{algorithm,algorithmic}
\usepackage{enumerate}
\usepackage[margin=25mm]{geometry}
\usepackage[numbers]{natbib}
\usepackage{mathtools}
\mathtoolsset{showonlyrefs}

\newcommand{\argmin}{\mathop{\rm argmin}\limits}
\newcommand{\dom}{\mathop{\rm dom}}
\newcommand{\prox}{\mathrm{prox}}

\newcommand{\sign}{\mathrm{sign}}
\newcommand{\innerprod}[2]{\langle #1,#2 \rangle}

\newtheorem{theorem}{Theorem}[section]
\newtheorem{corollary}{Corollary}[section]
\newtheorem{lemma}{Lemma}[section]

\newtheorem{proposition}{Proposition}[section]

\title{Fast algorithm for sparse least trimmed squares via trimmed-regularized reformulation}
\author{Shotaro Yagishita\thanks{The Institute of Statistical Mathematics, Japan, E-mail: syagi@ism.ac.jp}}
\date{\today}

\begin{document}

\maketitle

\begin{abstract}
The least trimmed squares (LTS) is a reasonable formulation of robust regression whereas it suffers from high computational cost due to the nonconvexity and nonsmoothness of its objective function.
The most frequently used FAST-LTS algorithm is particularly slow when a sparsity-inducing penalty such as the $\ell_1$ norm is added.
This paper proposes a computationally inexpensive algorithm for the sparse LTS, which is based on the proximal gradient method with a reformulation technique.
Proposed method is equipped with theoretical convergence preferred over existing methods.
Numerical experiments show that our method efficiently yields small objective value.
\end{abstract}

\section{Introduction}\label{sec:intro}
Outliers sometimes appear in data sets.
Nowadays, there exist several ways to remove the effect of outliers.
For regression problems, methods based on robust loss functions such as the Huber loss \citep{huber1964robust} and the Tukey loss \citep{beaton1974fitting} are widely used.
\citet{rousseeuw1984least} proposed the least trimmed squares (LTS) estimation, which minimizes the smallest partial sum of squares of the residuals instead of the sum of squares of the residuals and is equipped with high breakdown point.
Let $y=(y_1,\ldots,y_n)$ be the vector of responses and $X=(x_1,\ldots,x_n)^\top$ be the matrix of $d$-dimensional covariates $x_1,\ldots,x_n$.
The LTS for linear regression model defined as
\begin{equation}\label{problem:LTS}
    \underset{\beta_0\in\mathbb{R},\beta\in\mathbb{R}^d}{\mbox{minimize}} \quad \frac{1}{4}T_h(y-\beta_0\bm{1}-X\beta),
\end{equation}
where $1\le h\le n$, $\bm{1}$ is the all-ones vector,
\begin{equation}
    T_h(r)\coloneqq r_{\langle1\rangle}^2+\cdots+r_{\langle h\rangle}^2=\min_{\substack{\Lambda\subset[n]\\|\Lambda|=h}}\sum_{i\in\Lambda}r_i^2,
\end{equation}
and $|r_{\langle 1\rangle}|\le|r_{\langle 2\rangle}|\le\cdots\le|r_{\langle n\rangle}|$.
The function $T_h$ is called a trimmed squares function.
The loss function of the LTS is a reasonable robust one of ignoring $n-h$ worst-fitting samples.
To tackle high-dimensional data, \citet{alfons2013sparse} proposed the combination of the LTS and the least absolute shrinkage and selection operator (LASSO) \citep{tibshirani1996regression} called the sparse LTS (SLTS):
\begin{equation}\label{problem:SLTS}
    \underset{\beta_0\in\mathbb{R},\beta\in\mathbb{R}^d}{\mbox{minimize}} \quad \frac{1}{4}T_h(y-\beta_0\bm{1}-X\beta) + \lambda\|\beta\|_1,
\end{equation}
where $\lambda>0$ and $\|\cdot\|_1$ is the $\ell_1$ norm.

The LTS and SLTS are a reasonable formulation of robust regression whereas it suffers from high computational cost due to the nonconvexity and nonsmoothness of its objective function.
The FAST-LTS algorithm was introduced by \citet{rousseeuw2006computing} to solve the LTS.
Based on the fact that the LTS is equivalently rewritten as
\begin{equation}
    \underset{\beta_0\in\mathbb{R},\beta\in\mathbb{R}^d,|\Lambda|=h}{\mbox{minimize}} \quad \frac{1}{4}\sum_{i\in\Lambda}(y_i-\beta_0-x_i^\top\beta)^2
\end{equation}
the FAST-LTS performs an alternating minimization with respect to $(\beta_0,\beta)$ and $\Lambda$.
In the update of $\Lambda$, the $h$ samples that best fit the current regression model are selected, and in the update of $(\beta_0,\beta)$, the ordinary least squares is applied to the samples in the subset $\Lambda$.
For the SLTS, \citet{alfons2013sparse} have developed an algorithm analogous to FAST-LTS, which we call a FAST-SLTS. 
Unlike the FAST-LTS, the update of $(\beta_0,\beta)$ in the FAST-SLTS requires solving the LASSO problem, resulting in the use of algorithms such as the proximal gradient method (PGM, see e.g. \citet{parikh2014proximal}), also called the iterative shrinkage thresholding algorithm (ISTA), at each iteration.
In other words, because of its double-loop structure, the FAST-SLTS is computationally expensive.
Recently, \citet{aravkin2020trimmed} have proposed the trimmed stochastic variance reduced gradient method (TSVRG).
Although the TSVRG is a single-loop algorithm, the stepsizes used in it is restricted to very small, so convergence is slow in practice.

This paper proposes a computationally inexpensive single-loop algorithm for the SLTS.
The SLTS \eqref{problem:SLTS} is first equivalently rewritten in the form of a regularized regression problem.
While the regularized regression problem is similar to that of \citet{she2011outlier}, our regularizer is not separable with respect to the variables.
The PGM can be efficiently applied to the regularized problem and, combined with a linesearch technique, avoids excessively small stepsizes such as the TSVRG.
Moreover, owing to favorable properties of the regularized problem, we can show the linear convergence to a local minimizer, which is a stronger result than the convergence of the PGM for general nonconvex nonsmooth optimization problems.
Unlike the TSVRG, our reformulation technique can be easily extended to nonlinear regression.

The rest of this paper is organized as follows.
The remainder of this section is devoted to notation and preliminary results.
The next section introduces a regularized reformulation for the SLTS.
In Section \ref{sec:PGM}, we presents the PGM for the SLTS and its convergence analysis.
An extension to nonlinear regression is also provided.
Numerical experiments to demonstrate the efficiency of our proposal are reported in Section \ref{sec:experiments}.
Finally, Section \ref{sec:conclusion} concludes the paper.

\subsection{Notation and Preliminaries}
For an integer $n$, the set $[n]$ is defined by $[n]\coloneqq\{1,\ldots,n\}$.
The standard inner product and $\ell_2$ norm are denoted by $\innerprod{\cdot}{\cdot}$ and $\|\cdot\|_2$, respectively.
For a $(-\infty,\infty]$-valued function $\phi$, the Moreau envelope and proximal mapping at $r$ are defined by
\begin{align}
    M_{\phi}(r) &\coloneqq \inf_{\alpha}\Big\{\phi(\alpha)+\frac{1}{2}\|\alpha-r\|_2^2\Big\},\\
    \prox_{\phi}(r) &\coloneqq \argmin_{\alpha}\Big\{\phi(\alpha)+\frac{1}{2}\|\alpha-r\|_2^2\Big\},
\end{align}
respectively.
If the set $\prox_{\phi}(r)$ is a singleton, we denote its unique element by $\prox_{\phi}(r)$ for the convention.
It is well known that the proximal mapping of $c\|\cdot\|_1$ with $c>0$ is given by
\begin{equation}
    \big(\prox_{c\|\cdot\|_1}(r)\big)_i = \mathcal{S}_c(r_i) \coloneqq \sign(r_i)(|r_i|-c)_+,
\end{equation}
where $(\xi)_+\coloneqq\max\{\xi,0\}$ and $\sign(\xi)$ is $1$ if $\xi\ge0$ and $-1$ otherwise \citep[Example 6.8]{beck2017first}.

\section{Trimmed-regularized reformulation for SLTS}\label{sec:trimmed-reformulation}
To provide efficient algorithm for the SLTS, we first propose an equivalent reformulation of the SLTS.
The following proposition provides the Moreau envelope and proximal mapping of the trimmed squares function $T_h$, which plays a key role not only in the reformulation of the SLTS but also in constructing an efficient algorithm.

\begin{proposition}\label{prop:prox-trimmed-squares}
Let $\gamma>0,~ r\in\mathbb{R}^n$.
We define $\mathcal{I}$ as the set consisting of index sets $I$ of size $h$ such that $|r_i|\le|r_j|$ for any $i\in I,~ j\notin I$.
The set $\prox_{\gamma T_h}(r)$ consists of a vector $r^*$ such that
\begin{align}
    r^*_i=
    \begin{cases}
        \frac{r_i}{2\gamma+1}, &i\in I,\\
        r_i, &i\notin I
    \end{cases}
\end{align}
for some $I\in\mathcal{I}$.
Additionally, it holds that
\begin{equation}
    M_{\gamma T_h}(r)=\min_{\alpha\in\mathbb{R}^n}\Big\{\gamma T_h(\alpha)+\frac{1}{2}\|\alpha-r\|_2^2\Big\}=\frac{\gamma}{2\gamma+1}T_h(r).
\end{equation}
\end{proposition}

\begin{proof}
The minimization problem in $\prox_{\gamma T_h}(r)$ can be equivalently rewritten as
\begin{align}
    &\min_{\alpha\in\mathbb{R}^n}\bigg\{\gamma\min_{\substack{\Lambda\subset[n]\\|\Lambda|=h}}\sum_{i\in\Lambda}\alpha_i^2+\frac{1}{2}\sum_{i\in[n]}(\alpha_i-r_i)^2\bigg\}\\
    &=\min_{\substack{\Lambda\subset[n]\\|\Lambda|=h}}\min_{\alpha\in\mathbb{R}^n}\bigg\{\gamma\sum_{i\in\Lambda}\alpha_i^2+\frac{1}{2}\sum_{i\in[n]}(\alpha_i-r_i)^2\bigg\}\\
    &=\min_{\substack{\Lambda\subset[n]\\|\Lambda|=h}}\bigg\{\sum_{i\in\Lambda}\min_{\alpha_i\in\mathbb{R}}\left\{\gamma\alpha_i^2+\frac{1}{2}(\alpha_i-r_i)^2\right\}+\sum_{i\notin\Lambda}\min_{\alpha_i\in\mathbb{R}}\left\{\frac{1}{2}(\alpha_i-r_i)^2\right\}\bigg\}\\
    &=\min_{\substack{\Lambda\subset[n]\\|\Lambda|=h}}\bigg\{\sum_{i\in\Lambda}\frac{\gamma}{2\gamma+1}r_i^2\bigg\},
\end{align}
where the minimum value with respect to $\alpha$ in the last equality is attained at $\alpha_i=\frac{r_i}{2\gamma+1}$ for $i\in\Lambda$ and $\alpha_i=r_i$ for $i\notin\Lambda$.
Any index set $I\in\mathcal{I}$ minimize the last minimization problem, that is,
\begin{equation}
    \sum_{i\in I}\frac{\gamma}{2\gamma+1}r_i^2=\min_{\substack{\Lambda\subset[n]\\|\Lambda|=h}}\bigg\{\sum_{i\in\Lambda}\frac{\gamma}{2\gamma+1}r_i^2\bigg\}=\frac{\gamma}{2\gamma+1}T_h(r).
\end{equation}
This completes the proof.
\end{proof}

Proposition \ref{prop:prox-trimmed-squares} with $r=y-\beta_0\bm{1}-X\beta$ yields the equivalence of \eqref{problem:SLTS} and the following sparse trimmed-regularized least squares (STRLS):
\begin{equation}\label{problem:STRLS}
    \underset{\beta_0\in\mathbb{R},\beta\in\mathbb{R}^d,\alpha\in\mathbb{R}^n}{\mbox{minimize}} \quad \frac{1}{2}\|y-\beta_0\bm{1}-X\beta-\alpha\|_2^2 + \frac{1}{2}T_h(\alpha) + \lambda\|\beta\|_1 \eqqcolon L(\beta_0,\beta,\alpha)
\end{equation}

\begin{corollary}\label{cor:equivalence}
The following assertions holds.
\begin{enumerate}[(i)]
    \item Let $(\beta_0^*,\beta^*)$ be an optimal solution of \eqref{problem:SLTS}.
    Then, $(\beta_0^*,\beta^*,\prox_{\frac{1}{2}T_h}(y-\beta_0^*\bm{1}-X\beta^*))$ is optimal to \eqref{problem:STRLS}.
    \item Let $(\beta_0^*,\beta^*,\alpha^*)$ be an optimal solution of \eqref{problem:STRLS}.
    Then, $(\beta_0^*,\beta^*)$ is optimal to \eqref{problem:SLTS}.
\end{enumerate}
\end{corollary}

According to Corollary \ref{cor:equivalence}, we consider solving the STRLS instead of the SLTS.
Roughly speaking, since the STRLS succeeds in separating residuals $y-\beta_0\bm{1}-X\beta$ from the intractable nonconvexity and nonsmoothness of the trimmed squares function $T_h$, it is possible to apply the proximal gradient method.
Particularly, as we will see in the next section, the STRLS has desirable properties for applying the first-order methods.
Finally, we note that our reformulation technique is also applicable to nonlinear regression models $f_\beta$ and general nonconvex regularizer for $\beta$ because it does not depend on the form of the residuals with respect to the model parameter $\beta$ and the form of the regularization term.

\section{Proximal gradient method for STRLS}\label{sec:PGM}
Our proposed algorithm is based on the PGM.
Let $\eta_{\beta_0}^{(t)},\eta_\beta^{(t)},\eta_\alpha^{(t)}>0$, which correspond to the inverse of the stepsizes, the PGM repeats
\begin{align}
    &(\beta_0^{(t+1)},\beta^{(t+1)},\alpha^{(t+1)})\\
    &\in \argmin_{\beta_0,\beta,\alpha}\Big\{\innerprod{\nabla l(\beta_0^{(t)},\beta^{(t)},\alpha^{(t)})}{(\beta_0,\beta,\alpha)}+\frac{1}{2}\|(\beta_0,\beta,\alpha)-(\beta_0^{(t)},\beta^{(t)},\alpha^{(t)})\|_{\eta^{(t)}}^2 + \frac{1}{2}T_h(\alpha) + \lambda\|\beta\|_1\Big\},
\end{align}
where $l(\beta_0,\beta,\alpha)\coloneqq\frac{1}{2}\|y-\beta_0\bm{1}-X\beta-\alpha\|_2^2$, $\eta^{(t)}\coloneqq(\eta_{\beta_0}^{(t)},\eta_\beta^{(t)},\eta_\alpha^{(t)})$, and
\begin{equation}
    \|(\beta_0,\beta,\alpha)-(\beta_0^{(t)},\beta^{(t)},\alpha^{(t)})\|_{\eta^{(t)}} \coloneqq \Big\{\eta_{\beta_0}^{(t)}(\beta_0-\beta_0^{(t)})^2 + \eta_\beta^{(t)}\|(\beta-\beta^{(t)})\|_2^2 + \eta_\alpha^{(t)}\|\alpha-\alpha^{(t)}\|_2^2\Big\}^{\frac{1}{2}}.
\end{equation}
Namely, the PGM updates the sequence by the minimizer of the sum of linearization of the smooth term around the current iteration, quadratic regularization term, and nonsmooth terms.
From the separability of $\beta_0$, $\beta$, and $\alpha$ in the above subproblem, it follows from Proposition \ref{prop:prox-trimmed-squares} and the proximal mapping of the $\ell_1$ norm that
\begin{align}
    \beta_0^{(t+1)} &= \beta_0^{(t)} - \frac{1}{\eta_{\beta_0}^{(t)}}\nabla_{\beta_0}l(\beta_0^{(t)},\beta^{(t)},\alpha^{(t)}),\label{eq:intercept-update}\\
    \beta_j^{(t+1)} &= \mathcal{S}_\frac{\lambda}{\eta_{\beta_j}^{(t)}}\Big(\beta_j^{(t)} - \frac{1}{\eta_{\beta}^{(t)}}\nabla_{\beta_j}l(\beta_0^{(t)},\beta^{(t)},\alpha^{(t)})\Big) \quad (j=1,\ldots,d),\label{eq:cofficient-update}\\
    \alpha^{(t+1)} &\in \prox_{\frac{1}{2\eta_\alpha^{(t)}}T_h}\Big(\alpha^{(t)} - \frac{1}{\eta_{\alpha}^{(t)}}\nabla_\alpha l(\beta_0^{(t)},\beta^{(t)},\alpha^{(t)})\Big),\label{eq:absorber-update}
\end{align}
where $\nabla_{\beta_j}l$ and $\nabla_\alpha l$ are the gradient with respect to $\beta_j$ and $\alpha$, respectively.
The precise algorithm is described in Algorithm \ref{alg:PGM}.

\begin{algorithm}[H]
\caption{PGM for STRLS}
    \label{alg:PGM}
    \begin{algorithmic}
    \STATE {\bfseries Input:} $(\beta_0^{(0)},\beta^{(0)},\alpha^{(0)}),~ c_1>1,~ 0<c_2<1,~ \overline{\eta}\ge\underline{\eta}>0$, and $t=0$.
    \REPEAT
    \STATE Choose $\eta_{\beta_0}^{(t)},\eta_\beta^{(t)},\eta_\alpha^{(t)} \in [\underline{\eta},\overline{\eta}]$.
    \STATE Compute $(\beta_0^{(t+1)},\beta^{(t+1)},\alpha^{(t+1)})$ by \eqref{eq:intercept-update}, \eqref{eq:cofficient-update}, and \eqref{eq:absorber-update}.
    \WHILE{$L(\beta_0^{(t+1)},\beta^{(t+1)},\alpha^{(t+1)})>L(\beta_0^{(t)},\beta^{(t)},\alpha^{(t)})-\frac{c_2}{2}\|(\beta_0^{(t+1)},\beta^{(t+1)},\alpha^{(t+1)})-(\beta_0^{(t)},\beta^{(t)},\alpha^{(t)})\|_{\eta^{(t)}}^2$}
    \STATE Set $\eta_{\beta_0}^{(t)}\leftarrow c_1\eta_{\beta_0}^{(t)}$, $\eta_{\beta}^{(t)}\leftarrow c_1\eta_{\beta}^{(t)}$, and $\eta_\alpha^{(t)}\leftarrow c_1\eta_\alpha^{(t)}$.
    \STATE Compute $(\beta_0^{(t+1)},\beta^{(t+1)},\alpha^{(t+1)})$ by \eqref{eq:intercept-update}, \eqref{eq:cofficient-update}, and \eqref{eq:absorber-update}.
    \ENDWHILE
    \STATE Set $t\leftarrow t+1$.
    \UNTIL Terminated criterion is satisfied.
    \end{algorithmic}
\end{algorithm}

Note that the objective value is monotonically nonincreasing because of the acceptance criterion of the inverse of stepsizes.
The convergence result of the PGM for the STRLS is provided as follows.
The proof is displayed in Appendix \ref{sec:proofs}.

\begin{theorem}\label{thm:linear-convergence}
There exists a local minimum $(\beta_0^*,\beta^*,\alpha^*)$ of \eqref{problem:STRLS} such that $\{(\beta_0^{(t)},\beta^{(t)},\alpha^{(t)})\}$ converges to it.
In particular, there exist $C_1,C_2>0$ and $q_1,q_2\in(0,1)$ such that
\begin{align}
    L(\beta_0^{(t)},\beta^{(t)},\alpha^{(t)})-L(\beta_0^*,\beta^*,\alpha^*) &\le C_1q_1^t,\\
    \|(\beta_0^{(t)},\beta^{(t)},\alpha^{(t)})-(\beta_0^*,\beta^*,\alpha^*)\|_2 &\le C_2q_2^t.
\end{align}
\end{theorem}

Theorem \ref{thm:linear-convergence} shows that the convergence to a local minimizer of the STRLS, that is, it says that the sequence never attains any saddle point.
Typically, the PGM only guarantees the subsequential convergence to a stationary point.
In fact, the STRLS has no saddle points and thus the convergence to a locally optimal solution is obtained.
See Appendix \ref{sec:proofs} for the precise statement.
In addition, the sequence $\{(\beta_0^{(t)},\beta^{(t)},\alpha^{(t)})\}$ and its loss function values converge linearly.
Linear convergence is not guaranteed for general nonconvex optimization problems.
This desirable property also stems from the tractability of the STRLS.
Specifically, the objective function of the STRLS has the Kurdyka--\L ojasiewicz (KL) exponent of $1/2$ (see \citet{li2018calculus} for the definition), which is known to serve linear convergence for many first-order methods.
The TSVRG does not guarantee convergence to a locally optimal solution and only guarantees sublinear convergence of a certain optimality measure \citep{aravkin2020trimmed}.
Although \citet{aravkin2020trimmed} shows linear convergence of the TSVRG assuming a certain error bound, we are not sure if the establishment of the error bound can be validated.
The FAST-LTS is only guaranteed monotonicity of the loss function values.
Therefore, our proposed method has more favorable convergence properties.

\subsection{Extension to nonlinear models and general regularizer}
Let us consider the algorithmic aspect of the trimmed squares estimation of nonlinear models by the trimmed-regularized reformulation technique.
Specifically, the following trimmed-regularized problem is considered:
\begin{equation}\label{problem:general-TRLS}
    \underset{\beta\in\mathbb{R}^p,\alpha\in\mathbb{R}^n}{\mbox{minimize}} \quad \underbrace{\frac{1}{2}\|y-f_\beta(X)-\alpha\|_2^2}_{\eqqcolon \Tilde{l}(\beta,\alpha)} + \frac{1}{2}T_h(\alpha) + \lambda P(\beta) \eqqcolon \Tilde{L}(\beta,\alpha),
\end{equation}
where $P:\mathbb{R}^d\to(-\infty,\infty]$ is a penalty of some kind and $f_\beta(X)$ denotes $(f_\beta(x_1),\ldots,f_\beta(x_n))^\top$ for the convention.
Assuming that the mapping $\beta\mapsto f_\beta(X)$ is continuously differentiable and the proximal mapping of $cP$ can be easily computed for any $c>0$, each iteration of the PGM
\begin{align}
    \beta^{(t+1)} &\in \prox_{\frac{\lambda}{\eta_{\beta}^{(t)}}P}\Big(\beta^{(t)} - \frac{1}{\eta_{\beta}^{(t)}}\nabla_\beta \Tilde{l}(\beta^{(t)},\alpha^{(t)})\Big),\label{eq:cofficient-update-nonlinear}\\
    \alpha^{(t+1)} &\in \prox_{\frac{1}{2\eta_\alpha^{(t)}}T_h}\Big(\alpha^{(t)} - \frac{1}{\eta_{\alpha}^{(t)}}\nabla_\alpha \Tilde{l}(\beta^{(t)},\alpha^{(t)})\Big)\label{eq:absorber-update-nonlinear}
\end{align}
can be efficiently performed for \eqref{problem:general-TRLS}.
The algorithm is summarized in Algorithm \ref{alg:PGM-nonlinear}.

\begin{algorithm}[H]
\caption{PGM for \eqref{problem:general-TRLS}}
    \label{alg:PGM-nonlinear}
    \begin{algorithmic}
    \STATE {\bfseries Input:} $(\beta^{(0)},\alpha^{(0)}),~ c_1>1,~ 0<c_2<1,~ \overline{\eta}\ge\underline{\eta}>0$, and $t=0$.
    \REPEAT
    \STATE Choose $\eta_\beta^{(t)},\eta_\alpha^{(t)} \in [\underline{\eta},\overline{\eta}]$.
    \STATE Compute $(\beta^{(t+1)},\alpha^{(t+1)})$ by \eqref{eq:cofficient-update-nonlinear} and \eqref{eq:absorber-update-nonlinear}.
    \WHILE{$\Tilde{L}(\beta^{(t+1)},\alpha^{(t+1)})>\Tilde{L}(\beta^{(t)},\alpha^{(t)})-\frac{c_2}{2}\|(\beta^{(t+1)},\alpha^{(t+1)})-(\beta^{(t)},\alpha^{(t)})\|_{\eta^{(t)}}^2$}
    \STATE Set $\eta_{\beta}^{(t)}\leftarrow c_1\eta_{\beta}^{(t)}$ and $\eta_\alpha^{(t)}\leftarrow c_1\eta_\alpha^{(t)}$.
    \STATE Compute $(\beta^{(t+1)},\alpha^{(t+1)})$ by \eqref{eq:cofficient-update-nonlinear} and \eqref{eq:absorber-update-nonlinear}.
    \ENDWHILE
    \STATE Set $t\leftarrow t+1$.
    \UNTIL Terminated criterion is satisfied.
    \end{algorithmic}
\end{algorithm}

Unlike the linear model setting, Lipschitz continuity does not necessarily hold for $\nabla\Tilde{l}$.
The TSVRG is not applicable in such situation.
On the other hand, although the result is not as strong as in the linear model setting, the subsequential convergence to l-stationary points can be established as follows.

\begin{theorem}\label{thm:subsequential-convergence}
Assume the following:
\begin{enumerate}[(i)]
    \item The mapping $\beta\mapsto f_\beta(X)$ is continuously differentiable;
    \item The penalty $P$ is bounded from below;
    \item The penalty $P$ is lower semicontinuous;
    \item The penalty $P$ is continuous on $\dom P\coloneqq\{\beta\mid P(\beta)<\infty\}$.
\end{enumerate}
Then, either there exist $t$ such that $(\beta^{(t)},\alpha^{(t)})$ is a l-stationary point of \eqref{problem:general-TRLS} or any accumulation point of $\{(\beta^{(t)},\alpha^{(t)})\}$ is a l-stationary point of \eqref{problem:general-TRLS}.
\end{theorem}

See Appendix \ref{sec:proofs} for the definition of the l-stationarity and the proof.
Theorem \ref{thm:subsequential-convergence} is obtained from an extension of Theorem 3.1 in \citet{kanzow2022convergence}.

\section{Numerical examples}\label{sec:experiments}
To show the efficiency of our approach for the SLTS \eqref{problem:SLTS}, numerical experiments are conducted.
Throughout the experiments, the covariates $\Tilde{x}_1,\ldots,\Tilde{x}_n$ were independently generated from $N(0,\Sigma)$ with $\Sigma_{ij}=0.5^{|i-j|}$.
The intercept $\Tilde{\beta_0}$ and elements of coefficient vector $\Tilde{\beta}$ were drawn from $N(0,1)$ and each element of $\Tilde{\beta}$ was replaced by $0$ with probability $0.1$.
The response $\Tilde{y}_i$ followed
\begin{equation}
    \Tilde{y}_i = \Tilde{\beta_0} + \Tilde{x}_i^\top\Tilde{\beta} + e_i,
\end{equation}
where $e_i\sim N(20,2)$ for the $10\%$ sample and $e_i\sim N(0,1)$ for the rest.
Then, by using the median and median absolute deviation, $\Tilde{y}$ was centered and $\Tilde{X}$ was standardized to construct $y$ and $X$, respectively.
We took $h=\lfloor 0.75n \rfloor$.
As for setting of Algorithm \ref{alg:PGM}, $c_1=2,~ c_2=10^{-4},~ \underline{\eta}=10^{-10} ,~ \overline{\eta}=10^{10}$ were used.
To determine the initial of the inverse of the stepsizes, we used the following variable-wise variant of the Barzilai-Borwein initialization rule \citep{barzilai1988two}:
\begin{align}\label{eq:vwBB}
\eta_{\beta_0}^{(t)} &= \min\left\{\overline{\eta}, \max\left\{\underline{\eta},\frac{|\nabla_{\beta_0}l(\beta_0^{(t)},\beta^{(t)},\alpha^{(t)})-\nabla_{\beta_0}l(\beta_0^{(t-1)},\beta^{(t-1)},\alpha^{(t-1)})|}{|\beta_0^{(t)}-\beta_0^{(t-1)}|}\right\}\right\},\\
\eta_\beta^{(t)} &= \min\left\{\overline{\eta}, \max\left\{\underline{\eta},\frac{\innerprod{\nabla_\beta l(\beta_0^{(t)},\beta^{(t)},\alpha^{(t)})-\nabla_\beta l(\beta_0^{(t-1)},\beta^{(t-1)},\alpha^{(t-1)})}{\beta^{(t)}-\beta^{(t-1)}}}{\|\beta^{(t)}-\beta^{(t-1)}\|_2^2}\right\}\right\},\\
\eta_\alpha^{(t)} &= \min\left\{\overline{\eta}, \max\left\{\underline{\eta},\frac{\innerprod{\nabla_\alpha l(\beta_0^{(t)},\beta^{(t)},\alpha^{(t)})-\nabla_\alpha l(\beta_0^{(t-1)},\beta^{(t-1)},\alpha^{(t-1)})}{\alpha^{(t)}-\alpha^{(t-1)}}}{\|\alpha^{(t)}-\alpha^{(t-1)}\|_2^2}\right\}\right\}.
\end{align}
For $t=0$, $\eta_{\beta_0}^{(0)},\eta_\beta^{(0)},\eta_\alpha^{(0)}$ were set as $1$.
Algorithm \ref{alg:PGM} is terminated when the terminated criterion
\begin{equation}
    \|w^{(t)}\|_2 \le \|\nabla l(\beta_0^{(0)},\beta^{(0)},\alpha^{(0)})\|_2\times10^{-6}
\end{equation}
was satisfied or the number of iterations reacheed the maximum number of iterations $t_{\max}$, where
\begin{align}
    w^{(t)} &\coloneqq \nabla l(\beta_0^{(t)},\beta^{(t)},\alpha^{(t)}) - \nabla l(\beta_0^{(t-1)},\beta^{(t-1)},\alpha^{(t-1)})\\
    &\qquad - (\eta_{\beta_0}^{(t-1)}(\beta_0^{(t)}-\beta_0^{(t-1)}),\eta_\beta^{(t-1)}(\beta^{(t)}-\beta^{(t-1)}),\eta_\alpha^{(t-1)}(\alpha^{(t)}-\alpha^{(t-1)})).
\end{align}
Similar terminated criteria were used for the TSVRG and the PGM in inner iterations of the FAST-SLTS.
The objective function value of \eqref{problem:STRLS} is used in the evaluation criterion.
All the algorithms were implemented in MATLAB R2023b, and all the computations were conducted on a PC with OS: Windows CPU: 2.60 GHz and 16.0 GB memory.

\subsection{Comparison with TSVRG}
The proposed method is first compared with the TSVRG.
The maximum number of iterations was set as $t_{\max}=10^6$.
Each element of $(\beta_0^{(0)},\beta^{(0)})$ was generated from $N(0,1)$ and $\alpha^{(0)}$ was determined as $\alpha^{(t)}\in\argmin_\alpha L(\beta_0^{(0)},\beta^{(0)},\alpha)=\prox_{\frac{1}{2}T_h}(y-\beta_0^{(0)}\bm{1}-X\beta^{(0)})$. 
Convergence behaviors of the PGM and the TSVRG are shown in Figure \ref{fig:vs-TSVRG}.
While the PGM satisfied the terminated criterion before reaching the maximum number of iterations, the TSVRG did not.
Obviously, convergence of the TSVRG was quite slow.
This is because the stepsizes used in the TSVRG is restricted to very small.
We see from Figure \ref{fig:vs-TSVRG} (a) and (c) that the computation time of the PGM is decreasing in contrast to the increase in the dimension of variables.
We consider that this is due to the fact that the contribution of the improvement in the condition number of $X^\top X$ exceeds that of the increase in the number of variables.
This phenomenon is also observed in the next experiment.

\begin{figure}[H]
    \centering
    \begin{tabular}{cc}
        (a) $n=100,~ d=200$ & (b) $n=100,~ d=1000$ \\
        \begin{minipage}[t]{0.45\linewidth}
        \centering
        \includegraphics[width=1.0\columnwidth]{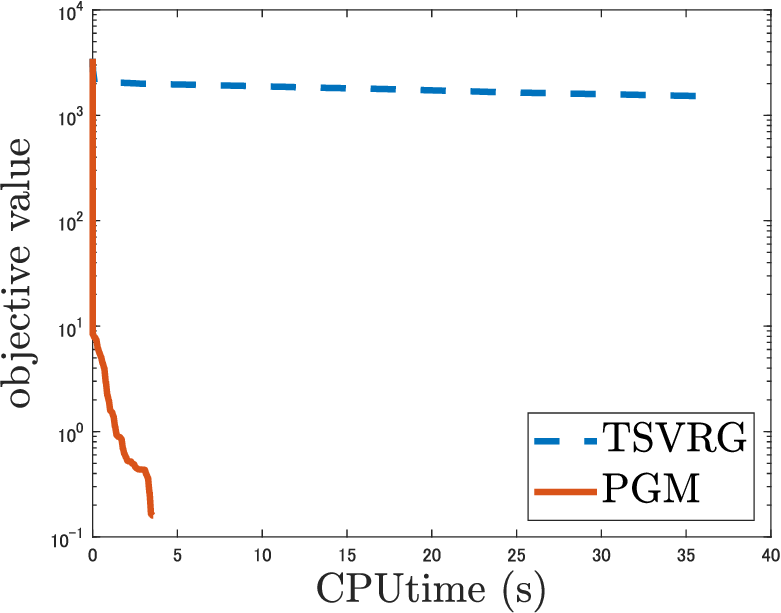}
        \end{minipage} &
        \begin{minipage}[t]{0.45\linewidth}
        \centering
        \includegraphics[width=1.0\columnwidth]{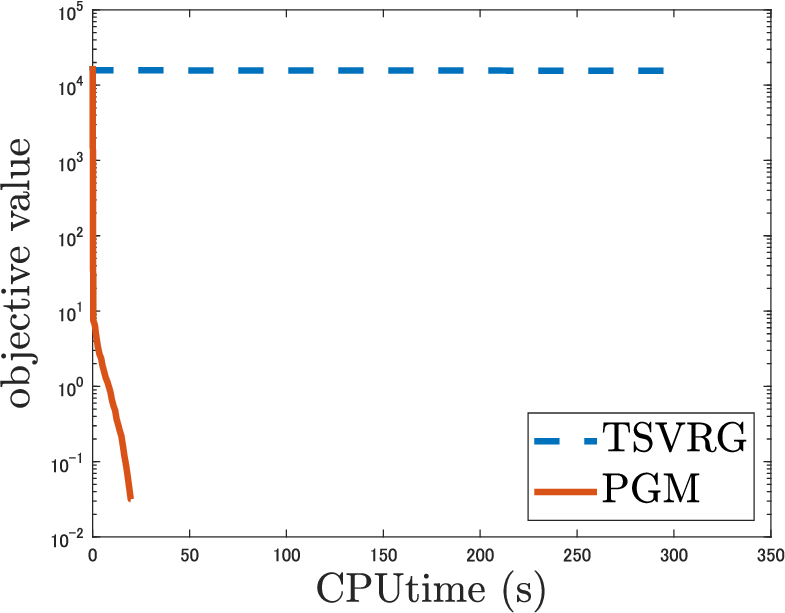}
        \end{minipage} \\
        (c) $n=500,~ d=200$ & (d) $n=500,~ d=1000$ \\
        \begin{minipage}[t]{0.45\linewidth}
        \centering
        \includegraphics[width=1.0\columnwidth]{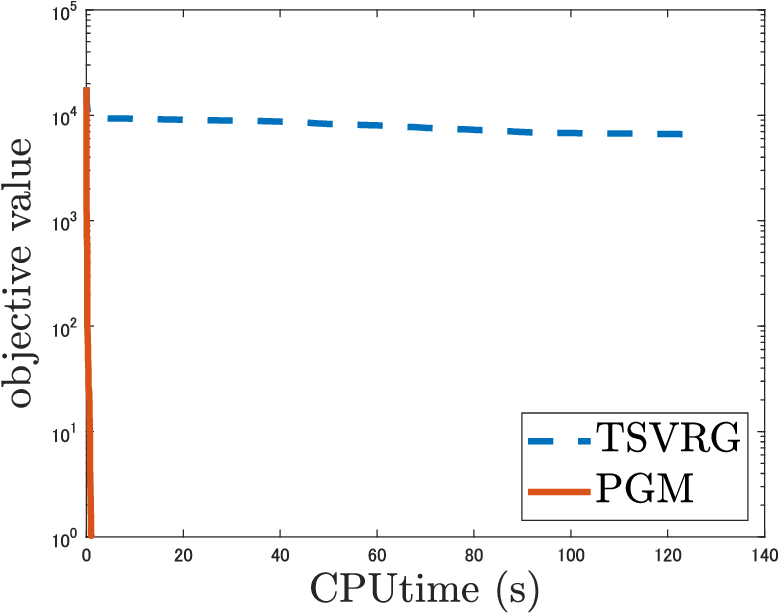}
        \end{minipage} &
        \begin{minipage}[t]{0.45\linewidth}
        \centering
        \includegraphics[width=1.0\columnwidth]{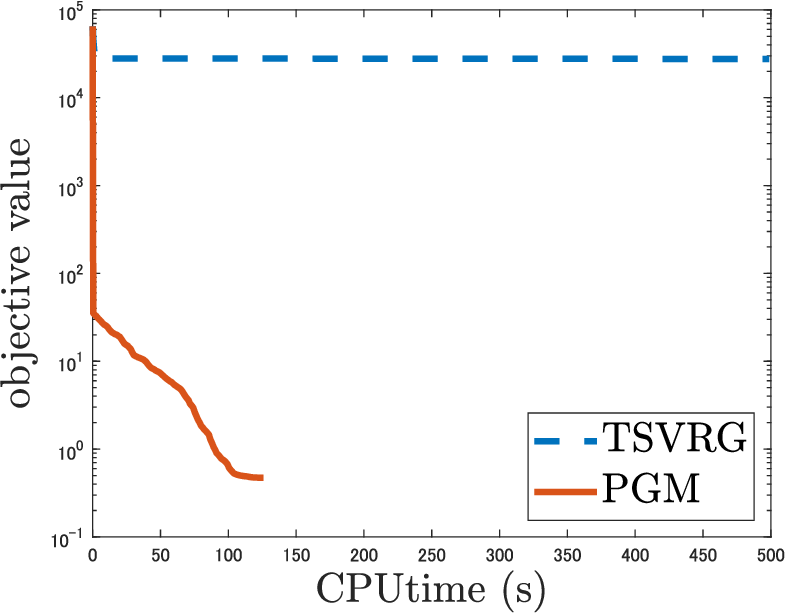}
        \end{minipage} \\
    \end{tabular}
    \caption{Convergence behaviors of the PGM and TSVRG.}
    \label{fig:vs-TSVRG}
\end{figure}

\subsection{Computational cost of PGM}
This subsection provides the variation in computational cost of the PGM with increasing the variables.
The maximum number of iterations was set as $t_{\max}=10^6$.
To reduce the influence of outliers, $3$ samples were randomly chosen and the solution of the LASSO for them was employed in the initial point $(\beta_0^{(0)},\beta^{(0)})$.
The initial point $\alpha^{(0)}$ was the same as in the previous subsection.
The average CPUtime for $10$ repetitions is displayed in Figure \ref{fig:time} and error bar indicates the standard deviation.
Figure \ref{fig:time} (a) shows that the computational cost grows linearly when the number of variables is increased while keeping the ratio of $n$ to $d$ constant.
On the other hand, we see from Figure \ref{fig:time} (b) that the computational cost increases nonlinearly when the dimension of the variables is increased while $n$ is kept constant.
Figure \ref{fig:time} (c) implies that the contribution of the condition number of $X^\top X$ is greater than that of the number of variables.

\begin{figure}[H]
    \begin{center}
    \begin{tabular}{c}
        (a) $d=2n$ \\
        \includegraphics[width=0.7\columnwidth]{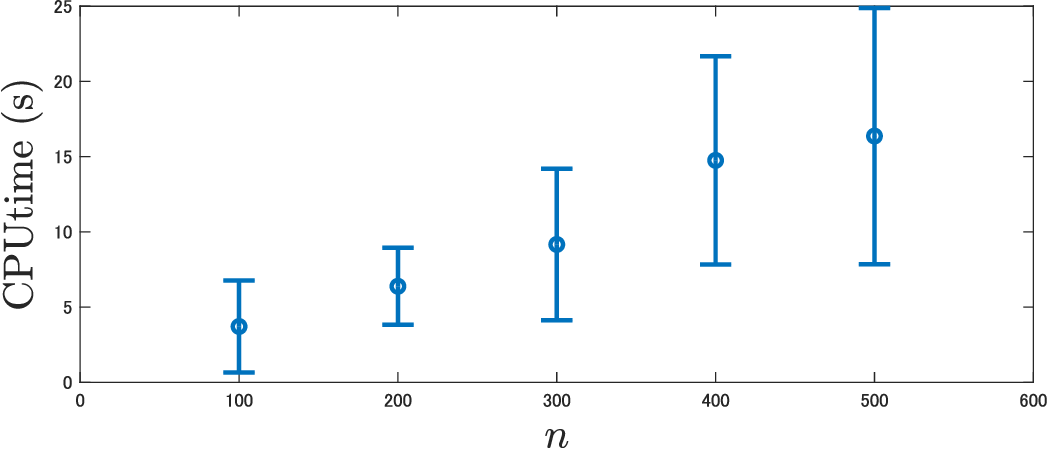} \\
        (b) $n=100$ \\
        \includegraphics[width=0.7\columnwidth]{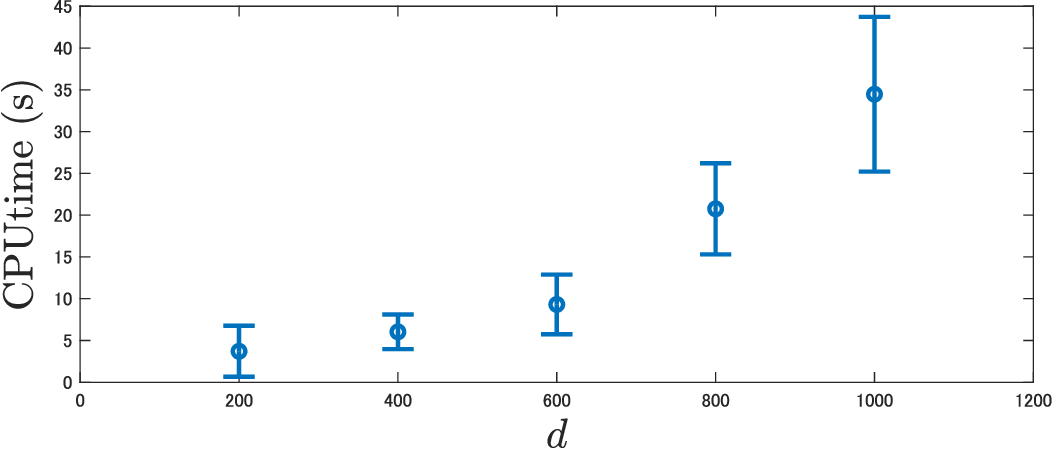} \\
        (c) $d=200$ \\
        \includegraphics[width=0.7\columnwidth]{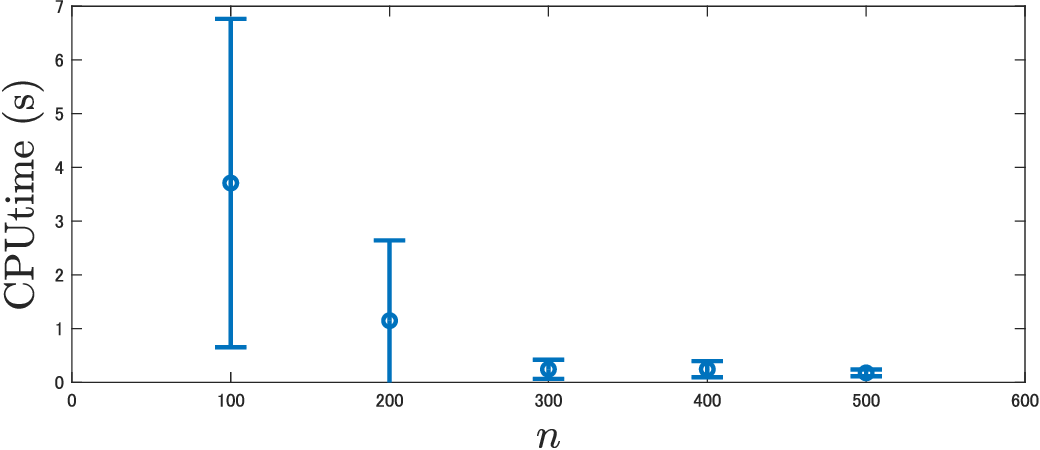} \\
    \end{tabular}
    \caption{Variation in CPUtime of the PGM due to increase in the number of variables.}
    \label{fig:time}
    \end{center}
\end{figure}

\subsection{Comparison with FAST-SLTS}
Since the SLTS is a nonconvex optimization problem, for practical purposes the algorithm is run from several initial points and the best solution is selected from among their outputs.
\citet{alfons2013sparse} proposed to run only $2$ iterations of the FAST-SLTS from $500$ initial points and then leave only $10$ good solutions from which to restart the iterations.
This subsection runs the PGM from several initial points and compares it to the FAST-SLTS.
The choice of initial points was the same as in the previous subsection and $t_{\max}=10^5$.
We compare our method with the FAST-SLTS by
\begin{align}
    \textbf{CPUtime ratio} &\coloneqq \frac{\textbf{CPUtime of PGM}}{\textbf{CPUtime of FAST-SLTS}},\\
    \textbf{objective value ratio} &\coloneqq \frac{\textbf{objective value of PGM}}{\textbf{objective value of FAST-SLTS}}.
\end{align}
The geometric mean, minimum, and maximum of the CPUtime ratio and objective value ratio for $10$ repetitions is summarized in Table \ref{table:vs-FAST-SLTS}.
With more than $5$ initial points, the objective value obtained by the PGM was only at most $2\%$ different from that obtained by the FAST-SLTS on average.
On the other hand, the computational cost of the PGM with $5$ initial points was about $5\%$ of that of the FAST-SLTS on average, indicating that solutions of similar quality can be obtained more quickly by our approach.

\begin{table}[H]
    \centering
    \caption{Summary of CPUtime ratio and objective value ratio.}
    \begin{tabular}{c||c|c}
        & CPUtime ratio & objective value ratio \\
        & geometric mean (min., max.) & geometric mean (min., max.) \\ \hline\hline
        $1$ initial point & 0.007 (0.002, 0.023) & 1.269 (1.007, 2.062) \\
        $5$ initial points & 0.048 (0.006, 0.101) & \textbf{1.018} (0.981, 1.071) \\
        $10$ initial points & 0.080 (0.012, 0.152) & \textbf{1.018} (0.983, 1.066) \\
        $20$ initial points & 0.154 (0.023, 0.402) & \textbf{1.007} (0.993, 1.030) \\
        $30$ initial points & 0.226 (0.029, 0.402) & \textbf{1.002} (0.986, 1.018) \\
    \end{tabular}
    \label{table:vs-FAST-SLTS}
\end{table}

\section{Concluding Remarks}\label{sec:conclusion}
This paper proposed the PGM-based algorithm for the trimmed-regularized reformulation of the SLTS.
Proposed method offers strong convergence properties.
In addition, it has the potential for expansion to nonlinear models.
Numerical experiments demonstrated the efficiency of our method.
Extension of our approach to the trimmed likelihood estimation \citep{hadi1997maximum} is an important task for the future.

\section*{Acknowledgments}
We would like to thank Akifumi Okuno and Hironori Fujisawa for the their valuable comments.

\appendix

\vspace{2em}
{\LARGE\bf\noindent Appendix}

\section{Proofs}\label{sec:proofs}
Before the proofs, some notions of stationarity are summarized.
Let $\phi$ be $(-\infty,\infty]$-valued function.
For $\Tilde{\xi}\in\dom\phi=\{\xi\mid\phi(\xi)<\infty\}$,
\begin{equation}
    \widehat{\partial}\phi(\Tilde{\xi}) \coloneqq \left\{g~\middle|~\liminf_{\xi\to\Tilde{\xi}}\frac{\phi(\xi)-\phi(\Tilde{\xi})-\innerprod{g}{\xi-\Tilde{\xi}}}{\|\xi-\Tilde{\xi}\|_2}\ge0\right\}
\end{equation}
is called the Fr\'echet subdifferential of $\phi$ at $\Tilde{\xi}$ and
\begin{equation}
    \partial\phi(\Tilde{\xi}) \coloneqq \left\{g~\middle|~\exists\{\xi^{(t)}\},\{g^{(t)}\}~ \mbox{s.t.}~ \xi^{(t)}\to\Tilde{\xi},~ \phi(\xi^{(t)})\to\phi(\Tilde{\xi}),~ g^{(t)}\to g,~ g^{(t)}\in\widehat{\partial}\phi(\Tilde{\xi}^{(t)})\right\}
\end{equation}
is known as the limiting subdifferential of $\phi$ at $\Tilde{\xi}$.
We call a point $\Tilde{\xi}\in\dom\phi$ satisfying $0\in\partial\phi(\Tilde{\xi})$ a l-stationary point of $\min_\xi\phi(\xi)$.
If $\phi$ is of the form $\phi=\phi_1+\phi_2$ where $\phi_1$ is continuously differentiable, it holds that $\widehat{\partial}\phi(\Tilde{\xi})=\nabla\phi_1(\Tilde{\xi})+\widehat{\partial}\phi_2(\Tilde{\xi})$ and $\partial\phi(\Tilde{\xi})=\nabla\phi_1(\Tilde{\xi})+\partial\phi_2(\Tilde{\xi})$ for $\Tilde{\xi}\in\dom\phi_2$ \citep[Exercise 8.8]{rockafellar2009variational}.
A real-valued function $\phi$ is directionally differentiable if the directional derivative of $\phi$ at $\xi$ in direction $d$
\begin{align}
\phi'(\xi;d)\coloneqq\lim_{\tau\searrow 0}\frac{\phi(\xi+\tau d)-\phi(\xi)}{\tau}
\end{align}
exists for any $\xi$ and $d$.
Note that the trimmed squares function $T_h$ is directionally differentiable \citep[Theorem 3.4]{delfour2019introduction}.
A point $\Tilde{\xi}$ is called a d-stationary point of $\min_\xi\phi(\xi)$ if $\phi'(\Tilde{\xi};d)\ge0$ holds for all $d$.
It is easy to see that any local minimizer $\xi^*$ of $\min_\xi\phi(\xi)$ satisfies $0\in\widehat{\partial}\phi(\xi^*)$ and $0\in\partial\phi(\xi^*)$, and is a d-stationary point.
See \citet{li2018calculus} for the definition of the KL function and the KL exponent.

\subsection{Proof of Theorem \ref{thm:linear-convergence}}
To prove Theorem \ref{thm:linear-convergence}, we present the following lemmas.

\begin{lemma}\label{lemma:boundedness-eta}
Let $\mathrm{Lip}_{\nabla l}>0$ be a Lipschitz constant of $\nabla l$.
It holds that
\begin{equation}
    \eta_{\beta_0}^{(t)},\eta_\beta^{(t)},\eta_\alpha^{(t)} \le \max\{c_1\mathrm{Lip}_{\nabla l}/(1-c_2),\overline{\eta}\}.
\end{equation}
\end{lemma}

\begin{proof}
Since $\nabla l$ is Lipschitz continuous, it follows from the well-known descent lemma \citep[Lemma 5.7]{beck2017first} that
\begin{align}
    l(\beta_0^{(t+1)},\beta^{(t+1)},\alpha^{(t+1)}) &\le l(\beta_0^{(t)},\beta^{(t)},\alpha^{(t)}) + \innerprod{\nabla l(\beta_0^{(t)},\beta^{(t)},\alpha^{(t)})}{(\beta_0^{(t+1)},\beta^{(t+1)},\alpha^{(t+1)})-(\beta_0^{(t)},\beta^{(t)},\alpha^{(t)})}\\
    &\qquad + \frac{\mathrm{Lip}_{\nabla l}}{2}\|(\beta_0^{(t+1)},\beta^{(t+1)},\alpha^{(t+1)})-(\beta_0^{(t)},\beta^{(t)},\alpha^{(t)})\|_2^2.
\end{align}
On the other hand, as $(\beta_0^{(t+1)},\beta^{(t+1)},\alpha^{(t+1)})$ is optimal to the subproblem, we have
\begin{align}
    &\innerprod{\nabla l(\beta_0^{(t)},\beta^{(t)},\alpha^{(t)})}{(\beta_0^{(t+1)},\beta^{(t+1)},\alpha^{(t+1)})} + \frac{1}{2}\|(\beta_0^{(t+1)},\beta^{(t+1)},\alpha^{(t+1)})-(\beta_0^{(t)},\beta^{(t)},\alpha^{(t)})\|_{\eta^{(t)}}^2\\
    &\qquad + \frac{1}{2}T_h(\alpha^{(t+1)}) + \lambda\|\beta^{(t+1)}\|_1\\
    &\le \innerprod{\nabla l(\beta_0^{(t)},\beta^{(t)},\alpha^{(t)})}{(\beta_0^{(t)},\beta^{(t)},\alpha^{(t)})} + \frac{1}{2}T_h(\alpha^{(t)}) + \lambda\|\beta^{(t)}\|_1.
\end{align}
Combining the rearranged above and the descent lemma yields
\begin{align}
    &L(\beta_0^{(t+1)},\beta^{(t+1)},\alpha^{(t+1)})\\
    &\le L(\beta_0^{(t)},\beta^{(t)},\alpha^{(t)}) + \frac{\mathrm{Lip}_{\nabla l}}{2}\|(\beta_0^{(t+1)},\beta^{(t+1)},\alpha^{(t+1)})-(\beta_0^{(t)},\beta^{(t)},\alpha^{(t)})\|_2^2\\
    &\qquad - \frac{1}{2}\|(\beta_0^{(t+1)},\beta^{(t+1)},\alpha^{(t+1)})-(\beta_0^{(t)},\beta^{(t)},\alpha^{(t)})\|_{\eta^{(t)}}^2\\
    &\le L(\beta_0^{(t)},\beta^{(t)},\alpha^{(t)}) - \frac{c_2}{2}\|(\beta_0^{(t+1)},\beta^{(t+1)},\alpha^{(t+1)})-(\beta_0^{(t)},\beta^{(t)},\alpha^{(t)})\|_{\eta^{(t)}}^2
\end{align}
provided that $\eta_{\beta_0}^{(t)},\eta_\beta^{(t)},\eta_\alpha^{(t)}\ge\mathrm{Lip}_{\nabla l}/(1-c_2)$ holds.
If the inner loop is executed at least once, the acceptance criterion does not holds for $(\eta_{\beta_0}^{(t)}/c_1,\eta_\beta^{(t)}/c_1,\eta_\alpha^{(t)}/c_1)$ and hence $\eta_{\beta_0}^{(t)},\eta_\beta^{(t)},\eta_\alpha^{(t)}<c_1\mathrm{Lip}_{\nabla l}/(1-c_2)$.
Otherwise, since the initial one is accepted, it holds that $\eta_{\beta_0}^{(t)},\eta_\beta^{(t)},\eta_\alpha^{(t)} \le \overline{\eta}$.
Consequently, we have the desired result.
\end{proof}

\begin{lemma}\label{lemma:level-boundedness}
For any $c\in\mathbb{R}$, the set $\mathcal{L}\coloneqq\{(\beta_0,\beta,\alpha)\mid L(\beta_0,\beta,\alpha)\le c\}$ is bounded.
\end{lemma}

\begin{proof}
Since it is obvious when $\mathcal{L}$ is an empty set, we consider only when $\mathcal{L}$ is nonempty.
In this case, $c\ge0$.
Letting
\begin{equation}\label{eq:L-Lambda}
    L_\Lambda(\beta_0,\beta,\alpha) \coloneqq \frac{1}{2}\|y-\beta_0\bm{1}-X\beta-\alpha\|_2^2 + \frac{1}{2}\sum_{i\in\Lambda}\alpha_i^2 + \lambda\|\beta\|_1,
\end{equation}
the objective function $L$ is represented as
\begin{equation}\label{eq:min-representation-TS}
    L(\beta_0,\beta,\alpha) = \min_{\substack{\Lambda\subset[n]\\|\Lambda|=h}}L_\Lambda(\beta_0,\beta,\alpha).
\end{equation}
Since it holds that $\mathcal{L}=\cup_{\substack{\Lambda\subset[n]\\|\Lambda|=h}}\mathcal{L}_\Lambda$ where $\mathcal{L}_\Lambda\coloneqq\{(\beta_0,\beta,\alpha)\mid L_\Lambda(\beta_0,\beta,\alpha)\le c\}$, it is sufficient to show the boundedness for any $\mathcal{L}_\Lambda$.
Let $(\beta_0^*,\beta^*,\alpha^*)\in\mathcal{L}_\Lambda$.
As each term on the right hand side of \eqref{eq:L-Lambda} is nonnegative, we have
\begin{align}
    \|y-\beta_0^*\bm{1}-X\beta^*-\alpha^*\|_2 &\le \sqrt{2c},\label{eq:residual-bound}\\
    \|\alpha_\Lambda^*\|_2 &\le \sqrt{2c},\label{eq:alpha-Lambda-bound}\\
    \|\beta^*\|_1 &\le c/\lambda,\label{eq:beta-bound}
\end{align}
where $\alpha_\Lambda$ denotes the subvector corresponding to $\Lambda$ for $\alpha$.
Letting $E_\Lambda$ and $E_{\Lambda^c}$ be the $n\times h$ and $n\times(n-h)$ matrices satisfying $\alpha=E_\Lambda\alpha_\Lambda+E_{\Lambda^c}\alpha_{\Lambda^c}$ for any $\alpha$, where $\alpha_{\Lambda^c}$ is the subvector of $\alpha$ corresponding to $\Lambda^c$, the linear mapping $(\beta_0,\alpha_{\Lambda^c})\mapsto\beta_0\bm{1}+E_{\Lambda^c}\alpha_{\Lambda^c}$ is injective and hence there exists $\mu>0$ such that
\begin{equation}\label{eq:bounded-below}
    \mu\|(\beta_0,\alpha_{\Lambda^c})\|_2 \le \|\beta_0\bm{1}+E_{\Lambda^c}\alpha_{\Lambda^c}\|_2
\end{equation}
for any $(\beta_0,\alpha_{\Lambda^c})$.
Combining \eqref{eq:bounded-below} with \eqref{eq:residual-bound}--\eqref{eq:beta-bound} yields
\begin{align}
    \mu\|(\beta_0^*,\alpha_{\Lambda^c}^*)\|_2 &\le \|\beta_0^*\bm{1}+E_{\Lambda^c}\alpha_{\Lambda^c}^*\|_2\\
    &= \|y-X\beta^*-E_\Lambda\alpha_\Lambda^*-(y-\beta_0^*\bm{1}-X\beta^*-\alpha^*)\|_2\\
    &\le \|y\|_2 + \|X\beta^*\|_2 + \|E_\Lambda\alpha_\Lambda^*\|_2 + \|y-\beta_0^*\bm{1}-X\beta^*-\alpha^*\|_2\\
    &\le \|y\|_2 + \|X\|_2\|\beta^*\|_2 + \|\alpha_\Lambda^*\|_2 + \|y-\beta_0^*\bm{1}-X\beta^*-\alpha^*\|_2\\
    &\le \|y\|_2 + c\|X\|_2/\lambda + \sqrt{2c} + \sqrt{2c},
\end{align}
where $\|X\|_2$ is the spectral norm of $X$.
The above, \eqref{eq:alpha-Lambda-bound}, and \eqref{eq:beta-bound} imply the boundedness of $\mathcal{L}_\Lambda$.
\end{proof}

\begin{lemma}\label{lemma:l-opt<=>d-stat}
Any d-stationary point of \eqref{problem:STRLS} is a local minimizer of \eqref{problem:STRLS}.
\end{lemma}

\begin{proof}
According to \eqref{eq:min-representation-TS}, the objective function $L$ is represented as the pointwise minimum of finitely many convex functions.
Applying Corollary 1 in \citep{yagishita2022exact} yields the desired result.
\end{proof}

\begin{lemma}\label{lemma:KL}
The objective function of \eqref{problem:STRLS} is a KL function and its KL exponent is $1/2$.
\end{lemma}

\begin{proof}
According to \eqref{eq:min-representation-TS}, the desired result follows from Corollary 5.2 in \citep{li2018calculus}.
\end{proof}

Lemma \ref{lemma:boundedness-eta} implies that the inner loop terminates after a finite number of iterations.
From Lemma \ref{lemma:l-opt<=>d-stat}, we see that the STRLS has no saddle points.

\begin{proof}[Proof of Theorem \ref{thm:linear-convergence}]
We first show that $\{(\beta_0^{(t)},\beta^{(t)},\alpha^{(t)})\}$ converges.
It follows from the acceptance criterion that
\begin{align}\label{eq:sufficient-decreasing}
\begin{split}
    L(\beta_0^{(t+1)},\beta^{(t+1)},\alpha^{(t+1)}) &\le L(\beta_0^{(t)},\beta^{(t)},\alpha^{(t)})-\frac{c_2}{2}\|(\beta_0^{(t+1)},\beta^{(t+1)},\alpha^{(t+1)})-(\beta_0^{(t)},\beta^{(t)},\alpha^{(t)})\|_{\eta^{(t)}}^2\\
    &\le L(\beta_0^{(t)},\beta^{(t)},\alpha^{(t)})-\frac{c_2\underline{\eta}}{2}\|(\beta_0^{(t+1)},\beta^{(t+1)},\alpha^{(t+1)})-(\beta_0^{(t)},\beta^{(t)},\alpha^{(t)})\|_2^2.
\end{split}
\end{align}
On the other hand, as $(\beta_0^{(t+1)},\beta^{(t+1)},\alpha^{(t+1)})$ is optimal to the subproblem, we have
\begin{equation}
    0 \in \nabla l(\beta_0^{(t)},\beta^{(t)},\alpha^{(t)}) + (\eta_{\beta_0}^{(t)}(\beta_0^{(t+1)}-\beta_0^{(t)}),\eta_\beta^{(t)}(\beta^{(t+1)}-\beta^{(t)}),\eta_\alpha^{(t)}(\alpha^{(t+1)}-\alpha^{(t)})) + \partial R(\beta_0^{(t+1)},\beta^{(t+1)},\alpha^{(t+1)}),
\end{equation}
where $R(\beta_0,\beta,\alpha)\coloneqq\frac{1}{2}T_h(\alpha)+\lambda\|\beta\|_1$.
Thus, the set
\begin{equation}
    \partial L(\beta_0^{(t+1)},\beta^{(t+1)},\alpha^{(t+1)})=\nabla l(\beta_0^{(t+1)},\beta^{(t+1)},\alpha^{(t+1)})+\partial R(\beta_0^{(t+1)},\beta^{(t+1)},\alpha^{(t+1)})
\end{equation}
contains
\begin{align}
    w^{(t+1)} &\coloneqq \nabla l(\beta_0^{(t+1)},\beta^{(t+1)},\alpha^{(t+1)}) - \nabla l(\beta_0^{(t)},\beta^{(t)},\alpha^{(t)})\\
    &\qquad - (\eta_{\beta_0}^{(t)}(\beta_0^{(t+1)}-\beta_0^{(t)}),\eta_\beta^{(t)}(\beta^{(t+1)}-\beta^{(t)}),\eta_\alpha^{(t)}(\alpha^{(t+1)}-\alpha^{(t)})).
\end{align}
From the Lipschitz continuity of $\nabla l$ and Lemma \ref{lemma:boundedness-eta}, we obtain
\begin{align}\label{eq:relative-error}
    \|w^{(t+1)}\|_2 \le \left(\mathrm{Lip}_{\nabla l}+\max\{c_1\mathrm{Lip}_{\nabla l}/(1-c_2),\overline{\eta}\}\right)\|(\beta_0^{(t+1)},\beta^{(t+1)},\alpha^{(t+1)})-(\beta_0^{(t)},\beta^{(t)},\alpha^{(t)})\|_2.
\end{align}
Since the objective function of \eqref{problem:STRLS} is a KL function according to Lemma \ref{lemma:KL}, if $\{(\beta_0^{(t)},\beta^{(t)},\alpha^{(t)})\}$ has an accumulation point, from Theorem 2.9 in \citep{attouch2013convergence} with \eqref{eq:sufficient-decreasing} and \eqref{eq:relative-error}, the whole sequence converges to the accumulation point.
The inequality \eqref{eq:sufficient-decreasing} implies that $\{L(\beta_0^{(t)},\beta^{(t)},\alpha^{(t)})\}$ is monotonically nonincreasing.
Consequently, $\{(\beta_0^{(t)},\beta^{(t)},\alpha^{(t)})\}\subset\{(\beta_0,\beta,\alpha)\mid L(\beta_0,\beta,\alpha)\le L(\beta_0^{(0)},\beta^{(0)},\alpha^{(0)})\}$ and hence, from Lemma \ref{lemma:level-boundedness}, the sequence is bounded.
Thus, $\{(\beta_0^{(t)},\beta^{(t)},\alpha^{(t)})\}$ has an accumulation point.
The limit point is denoted by $(\beta_0^*,\beta^*,\alpha^*)$.

Next, we show the local optimality of $(\beta_0^*,\beta^*,\alpha^*)$.
Since $(\beta_0^{(t+1)},\beta^{(t+1)},\alpha^{(t+1)})$ is optimal to the subproblem, it follows from Lemma \ref{lemma:boundedness-eta} that
\begin{align}
    &\innerprod{\nabla l(\beta_0^{(t)},\beta^{(t)},\alpha^{(t)})}{(\beta_0^{(t+1)},\beta^{(t+1)},\alpha^{(t+1)})} + \frac{\underline{\eta}}{2}\|(\beta_0^{(t+1)},\beta^{(t+1)},\alpha^{(t+1)})-(\beta_0^{(t)},\beta^{(t)},\alpha^{(t)})\|_2^2\\
    &\qquad + R(\beta_0^{(t+1)},\beta^{(t+1)},\alpha^{(t+1)})\\
    &\le\innerprod{\nabla l(\beta_0^{(t)},\beta^{(t)},\alpha^{(t)})}{(\beta_0^{(t+1)},\beta^{(t+1)},\alpha^{(t+1)})} + \frac{1}{2}\|(\beta_0^{(t+1)},\beta^{(t+1)},\alpha^{(t+1)})-(\beta_0^{(t)},\beta^{(t)},\alpha^{(t)})\|_{\eta^{(t)}}^2\\
    &\qquad + R(\beta_0^{(t+1)},\beta^{(t+1)},\alpha^{(t+1)})\\
    &\le \innerprod{\nabla l(\beta_0^{(t)},\beta^{(t)},\alpha^{(t)})}{(\beta_0^*,\beta^*,\alpha^*)+\tau d} + \frac{1}{2}\|(\beta_0^*,\beta^*,\alpha^*)+\tau d-(\beta_0^{(t)},\beta^{(t)},\alpha^{(t)})\|_{\eta^{(t)}}^2\\
    &\qquad + R((\beta_0^*,\beta^*,\alpha^*)+\tau d)\\
    &\le \innerprod{\nabla l(\beta_0^{(t)},\beta^{(t)},\alpha^{(t)})}{(\beta_0^*,\beta^*,\alpha^*)+\tau d} + \frac{\overline{\eta}}{2}\|(\beta_0^*,\beta^*,\alpha^*)+\tau d-(\beta_0^{(t)},\beta^{(t)},\alpha^{(t)})\|_2^2\\
    &\qquad + R((\beta_0^*,\beta^*,\alpha^*)+\tau d)
\end{align}
for all $d$ and $\tau$.
By taking the limit $t\to\infty$, we see from the continuity of $R$ that
\begin{equation}
    \tau\innerprod{\nabla l(\beta_0^*,\beta^*,\alpha^*)}{d} + \frac{\overline{\eta}\tau^2}{2}\|d\|_2^2 + R((\beta_0^*,\beta^*,\alpha^*)+\tau d) - R(\beta_0^*,\beta^*,\alpha^*) \ge 0
\end{equation}
Dividing both sides by $\tau$ and taking the limit $\eta\to0$ give
\begin{equation}
    L'((\beta_0^*,\beta^*,\alpha^*);d) = \innerprod{\nabla l(\beta_0^*,\beta^*,\alpha^*)}{d} + R'((\beta_0^*,\beta^*,\alpha^*);d) \ge 0,
\end{equation}
which implies that $(\beta_0^*,\beta^*,\alpha^*)$ is a d-stationary point of \eqref{problem:STRLS}.
According to Lemma \ref{lemma:l-opt<=>d-stat}, $(\beta_0^*,\beta^*,\alpha^*)$ is a local minimizer of \eqref{problem:STRLS}.

Lastly, since the KL exponent of the objective function of \eqref{problem:STRLS} is $1/2$ by Lemma \ref{lemma:KL}, applying Theorem 3.4 in \citep{frankel2015splitting} yields the linear convergence of $\{(\beta_0^{(t)},\beta^{(t)},\alpha^{(t)})\}$ and $\{L(\beta_0^{(t)},\beta^{(t)},\alpha^{(t)})\}$.
\end{proof}

\subsection{Proof of Theorem \ref{thm:subsequential-convergence}}
Although \citet{kanzow2022convergence} considers the subsequential convergence of the PGM, the standard $\ell_2$ norm is used rather than the weighted norm $\|\cdot\|_{\eta^{(t)}}$ there.
Fortunately, thanks to taking the initial values so that $\underline{\eta} \le \eta_\beta^{(t)},\eta_\alpha^{(t)} \le \overline{\eta}$, we obtain the following lemmas in the same way as in the proof of Lemma 3.1, Proposition 3.1, and Proposition 3.2 in \citet{kanzow2022convergence}.

\begin{lemma}\label{lemma:well-definedness}
Suppose the assumptions of Theorem \ref{thm:subsequential-convergence}.
If $(\beta^{(t)},\alpha^{(t)})$ is not an l-stationary point of \eqref{problem:general-TRLS}, the inner loop of Algorithm \ref{alg:PGM-nonlinear} at the iteration $t$ is finite.
\end{lemma}

\begin{lemma}\label{lemma:difference-diminishing}
Let the assumptions of Theorem \ref{thm:subsequential-convergence} hold and $(\beta^{(t)},\alpha^{(t)})$ be not an l-stationary point of \eqref{problem:general-TRLS} for all $t$.
Then, $\|(\beta^{(t+1)},\alpha^{(t+1)})-(\beta^{(t)},\alpha^{(t)})\|_2\to0$ holds.
\end{lemma}

\begin{lemma}\label{lemma:subsequential-diminishing}
Suppose that the assumptions of Theorem \ref{thm:subsequential-convergence} hold and $(\beta^{(t)},\alpha^{(t)})$ is not an l-stationary point of \eqref{problem:general-TRLS} for all $t$.
Let $\{(\beta^{(t)},\alpha^{(t)})\}_{t\in T}$ be a subsequence converging to a point $(\beta^*,\alpha^*)$.
Then, $\eta_\beta^{(t)}\|\beta^{(t+1)}-\beta^{(t)}\|_2\to0$ and $\eta_\alpha^{(t)}\|\alpha^{(t+1)}-\alpha^{(t)}\|_2\to0$.
\end{lemma}

Using Lemmas \ref{lemma:well-definedness}--\ref{lemma:subsequential-diminishing}, we prove Theorem \ref{thm:subsequential-convergence}.

\begin{proof}[Proof of Theorem \ref{thm:subsequential-convergence}]
Suppose that $(\beta^{(t)},\alpha^{(t)})$ is not an l-stationary point of \eqref{problem:general-TRLS} for all $t$.
Let $(\beta^*,\alpha^*)$ be an accumulation point of $\{(\beta^{(t)},\alpha^{(t)})\}$ and $\{(\beta^{(t)},\alpha^{(t)})\}_{t\in T}$ be a subsequence converging to the accumulation point.
In view of Lemma \ref{lemma:difference-diminishing}, it holds that $(\beta^{(t+1)},\alpha^{(t+1)})\to_T(\beta^*,\alpha^*)$.
Since $(\beta^{(t+1)},\alpha^{(t+1)})$ is a minimizer of the subproblem, we have
\begin{equation}
    0 \in \nabla\Tilde{l}(\beta^{(t)},\alpha^{(t)}) + (\eta_\beta^{(t)}(\beta^{(t+1)}-\beta^{(t)}),\eta_\alpha^{(t)}(\alpha^{(t+1)}-\alpha^{(t)})) + \widehat{\partial} \Tilde{R}(\beta^{(t+1)},\alpha^{(t+1)}),
\end{equation}
where $\Tilde{R}(\beta,\alpha)\coloneqq\frac{1}{2}T_h(\alpha)+\lambda P(\beta)$.
From Lemma \ref{lemma:subsequential-diminishing} and the continuity of $\nabla\Tilde{l}$, taking the partial limit $t\to_T\infty$ yields
\begin{equation}
    0 \in \nabla\Tilde{l}(\beta^*,\alpha^*) + \partial\Tilde{R}(\beta^*,\alpha^*) = \partial\Tilde{L}(\beta^*,\alpha^*),
\end{equation}
which implies that $(\beta^*,\alpha^*)$ is a l-stationary point of \eqref{problem:general-TRLS}.
\end{proof}

\bibliography{reference.bib}

\begin{thebibliography}{19}
\providecommand{\natexlab}[1]{#1}
\providecommand{\url}[1]{\texttt{#1}}
\expandafter\ifx\csname urlstyle\endcsname\relax
  \providecommand{\doi}[1]{doi: #1}\else
  \providecommand{\doi}{doi: \begingroup \urlstyle{rm}\Url}\fi

\bibitem[Alfons et~al.(2013)Alfons, Croux, and Gelper]{alfons2013sparse}
Andreas Alfons, Christophe Croux, and Sarah Gelper.
\newblock Sparse least trimmed squares regression for analyzing high-dimensional large data sets.
\newblock \emph{The Annals of Applied Statistics}, pages 226--248, 2013.

\bibitem[Aravkin and Davis(2020)]{aravkin2020trimmed}
Aleksandr Aravkin and Damek Davis.
\newblock Trimmed statistical estimation via variance reduction.
\newblock \emph{Mathematics of Operations Research}, 45\penalty0 (1):\penalty0 292--322, 2020.

\bibitem[Attouch et~al.(2013)Attouch, Bolte, and Svaiter]{attouch2013convergence}
Hedy Attouch, J{\'e}r{\^o}me Bolte, and Benar~Fux Svaiter.
\newblock Convergence of descent methods for semi-algebraic and tame problems: proximal algorithms, forward--backward splitting, and regularized gauss--seidel methods.
\newblock \emph{Mathematical Programming}, 137\penalty0 (1):\penalty0 91--129, 2013.

\bibitem[Barzilai and Borwein(1988)]{barzilai1988two}
Jonathan Barzilai and Jonathan~M Borwein.
\newblock Two-point step size gradient methods.
\newblock \emph{IMA journal of numerical analysis}, 8\penalty0 (1):\penalty0 141--148, 1988.

\bibitem[Beaton and Tukey(1974)]{beaton1974fitting}
Albert~E Beaton and John~W Tukey.
\newblock The fitting of power series, meaning polynomials, illustrated on band-spectroscopic data.
\newblock \emph{Technometrics}, 16\penalty0 (2):\penalty0 147--185, 1974.

\bibitem[Beck(2017)]{beck2017first}
Amir Beck.
\newblock \emph{First-order methods in optimization}.
\newblock SIAM, 2017.

\bibitem[Delfour(2019)]{delfour2019introduction}
Michel~C Delfour.
\newblock \emph{Introduction to Optimization and Hadamard Semidifferential Calculus}.
\newblock SIAM, 2019.

\bibitem[Frankel et~al.(2015)Frankel, Garrigos, and Peypouquet]{frankel2015splitting}
Pierre Frankel, Guillaume Garrigos, and Juan Peypouquet.
\newblock Splitting methods with variable metric for kurdyka--{\l}ojasiewicz functions and general convergence rates.
\newblock \emph{Journal of Optimization Theory and Applications}, 165:\penalty0 874--900, 2015.

\bibitem[Hadi and Luce{\~n}o(1997)]{hadi1997maximum}
Ali~S Hadi and Alberto Luce{\~n}o.
\newblock Maximum trimmed likelihood estimators: a unified approach, examples, and algorithms.
\newblock \emph{Computational Statistics \& Data Analysis}, 25\penalty0 (3):\penalty0 251--272, 1997.

\bibitem[Huber(1964)]{huber1964robust}
Peter~J. Huber.
\newblock {Robust Estimation of a Location Parameter}.
\newblock \emph{The Annals of Mathematical Statistics}, 35\penalty0 (1):\penalty0 73--101, 1964.

\bibitem[Kanzow and Mehlitz(2022)]{kanzow2022convergence}
Christian Kanzow and Patrick Mehlitz.
\newblock Convergence properties of monotone and nonmonotone proximal gradient methods revisited.
\newblock \emph{Journal of Optimization Theory and Applications}, 195\penalty0 (2):\penalty0 624--646, 2022.

\bibitem[Li and Pong(2018)]{li2018calculus}
Guoyin Li and Ting~Kei Pong.
\newblock Calculus of the exponent of kurdyka--{\l}ojasiewicz inequality and its applications to linear convergence of first-order methods.
\newblock \emph{Foundations of computational mathematics}, 18\penalty0 (5):\penalty0 1199--1232, 2018.

\bibitem[Parikh and Boyd(2014)]{parikh2014proximal}
Neal Parikh and Stephen Boyd.
\newblock Proximal algorithms.
\newblock \emph{Foundations and trends{\textregistered} in Optimization}, 1\penalty0 (3):\penalty0 127--239, 2014.

\bibitem[Rockafellar and Wets(2009)]{rockafellar2009variational}
R~Tyrrell Rockafellar and Roger J-B Wets.
\newblock \emph{Variational analysis}, volume 317.
\newblock Springer Science \& Business Media, 2009.

\bibitem[Rousseeuw(1984)]{rousseeuw1984least}
Peter~J Rousseeuw.
\newblock Least median of squares regression.
\newblock \emph{Journal of the American statistical association}, 79\penalty0 (388):\penalty0 871--880, 1984.

\bibitem[Rousseeuw and Van~Driessen(2006)]{rousseeuw2006computing}
Peter~J Rousseeuw and Katrien Van~Driessen.
\newblock Computing lts regression for large data sets.
\newblock \emph{Data mining and knowledge discovery}, 12:\penalty0 29--45, 2006.

\bibitem[She and Owen(2011)]{she2011outlier}
Yiyuan She and Art~B Owen.
\newblock Outlier detection using nonconvex penalized regression.
\newblock \emph{Journal of the American Statistical Association}, 106\penalty0 (494):\penalty0 626--639, 2011.

\bibitem[Tibshirani(1996)]{tibshirani1996regression}
Robert Tibshirani.
\newblock Regression shrinkage and selection via the lasso.
\newblock \emph{Journal of the Royal Statistical Society: Series B (Methodological)}, 58\penalty0 (1):\penalty0 267--288, 1996.

\bibitem[Yagishita and Gotoh(2022)]{yagishita2022exact}
Shotaro Yagishita and Jun-ya Gotoh.
\newblock Exact penalization at d-stationary points of cardinality-or rank-constrained problem.
\newblock \emph{arXiv preprint arXiv:2209.02315}, 2022.

\end{thebibliography}
\bibliographystyle{plainnat}

\end{document}